\documentclass{article}

\usepackage{amsmath,amssymb,graphicx,amsthm,ifthen,epstopdf,fancyhdr,color}
\usepackage[square,numbers,sort&compress]{natbib}
\usepackage{mathtools}
\usepackage{epstopdf}
\usepackage[normalem]{ulem}
\usepackage{caption}
\usepackage{subcaption}
\usepackage[T1]{fontenc}
\usepackage[utf8]{inputenc}
\usepackage{cases}

\theoremstyle{theorem} \newtheorem{thm}{Theorem}
\theoremstyle{definition} 
\theoremstyle{definition} \newtheorem{defn}{Definition}
\theoremstyle{definition} \newtheorem{lemma}{Lemma}
\theoremstyle{definition}

\usepackage{color}
\usepackage{ifthen}

  \newcommand{\R}{\ensuremath{\mathbb{R}}}
  \newcommand{\E}{\ensuremath{\mathbb{E}}}
  
  \newcommand{\Uc}{\mathcal{U}}
  \newcommand{\Vc}{\mathcal{V}}
  \newcommand{\Zc}{\mathcal{Z}}

  \newcommand{\Fc}{\mathcal{F}}
  
  \newcommand{\Bc}{\mathcal{B}}
  
  \newcommand{\bx}{\bar{x}}

\newcommand{\V}[1]{\ensuremath{\mathbf{#1}}}

\newcommand{\norm}[1]{\left|\left| #1 \right|\right|}

\newcommand{\TODO}[1]{ 
\ifx\NOTES\undefined\else
{\color{red} [!]}\footnote{ {\color{red} TODO: #1}}
\fi
}

\newcommand{\aslim}{\stackrel{a.s.}{\rightarrow}}
\newcommand{\aseq}{\stackrel{a.s.}{=}}

\newcommand{\iid}{\stackrel{\text{iid}}{\sim}}

\newcommand{\indep}{\stackrel{\text{indep.}}{\sim}}

\newcommand{\Pd}[3]{\ifthenelse{\equal{#3}{1}}{\frac{\partial #1}{\partial #2}}{\frac{\partial^{#3} #1}{\partial #2^{#3}}}}

\newcommand{\Vu}{\mathbf{u}} \newcommand{\Vut}{\mathbf{\tilde{u}}}
\newcommand{\Vv}{\mathbf{v}} \newcommand{\Vvt}{\mathbf{\tilde{v}}}

\newcommand{\Vt}[1]{\ensuremath{\mathbf{\tilde{#1}}}}

\usepackage[export]{adjustbox}

\usepackage[final]{nips_2017}

\usepackage[utf8]{inputenc} % allow utf-8 input
\usepackage[T1]{fontenc}    % use 8-bit T1 fonts
\usepackage{hyperref}       % hyperlinks
\usepackage{url}            % simple URL typesetting
\usepackage{booktabs}       % professional-quality tables
\usepackage{amsfonts}       % blackboard math symbols
\usepackage{nicefrac}       % compact symbols for 1/2, etc.
\usepackage{microtype}      % microtypography

\newcommand{\DANNY}[1]{ 
\ifx\NOTES\undefined\else
{\color{blue} [!]}\footnote{ {\color{blue} Review: #1}}
\fi
}

\title{Optimal Shrinkage of Singular Values Under Random Data Contamination}

\author{
  Danny Barash \\
  School of Computer Science and Engineering\\
  Hebrew University\\
  Jerusalem, Israel \\
  \texttt{danny.barash@mail.huji.ac.il} \\
  \And
  Matan Gavish \\
  School of Computer Science and Engineering\\
  Hebrew University\\
  Jerusalem, Israel \\
  \texttt{gavish@cs.huji.ac.il} \\
}

\begin{document}

\maketitle \begin{abstract}
  A low rank matrix $X$ has been contaminated by
  uniformly distributed noise, missing values, outliers and corrupt entries.
  Reconstruction of $X$ from the singular values and singular vectors of the
  contaminated matrix $Y$ is a key problem in machine learning, computer vision
  and data science.  In this paper, we show that common contamination models
  (including arbitrary combinations of uniform noise, missing values, outliers
  and corrupt entries) can be described efficiently using a single framework.
  We develop an asymptotically optimal algorithm that estimates $X$ by
  manipulation of the singular values of $Y$, which applies to any of the
  contamination models considered.  Finally, we find an explicit signal-to-noise
cutoff, below which estimation of $X$ from the singular value decomposition of
$Y$ must fail, in a well-defined sense.  
\end{abstract}

\section{Introduction}
Reconstruction of low-rank matrices from noisy and otherwise contaminated data
is a key problem in machine learning, computer vision and data science.
Well-studied problems
such as  dimension reduction \cite{boutsidis2015randomized}, collaborative
filtering \cite{rao2015collaborative,luo2014efficient}, topic models
\cite{das2015gaussian}, video processing \cite{Ji2010}, face recognition
\cite{Yang2017}, predicting preferences \cite{meloun2000critical}, analytical
chemistry \cite{Rennie2005} and background-foreground separation
\cite{Bouwmans2016} all reduce, under popular approaches, to low-rank matrix reconstruction.  A
significant  part of the literature on these problems is based on the
singular value decomposition (SVD) as the underlying algorithmic component, see
e.g. \cite{Candes2010,Hastie1999,lin2013}.

Understanding and improving the behavior of
SVD in the presence of random data contamination therefore arises as 
a crucially important problem in machine learning. 
While this is certainly a classical problem 
\cite{huber2011robust,gnanadesikan1972robust,fischler1981random},
it remains of significant
interest, owing in part to the emergence of low-rank matrix models for 
matrix completion and collaborative filtering
\cite{Wright2009,Candes2011}.

Let $X$ be an $m$-by-$n$ unknown low-rank matrix of interest ($m\leq n$), and assume that we
only observe the data matrix $Y$, which is a contaminated or noisy version 
 of $X$.
 Let 
\begin{align} \label{svd:eq}
  Y=\sum_{i=1}^m y_i \V{u}_i \V{v}_i'
\end{align}
be the SVD of the data matrix $Y$. 
Any algorithm based on the SVD  
essentially aims 
to obtain an estimate for the target matrix $X$ from \eqref{svd:eq}.
Most practitioners simply form the Truncated SVD (TSVD) estimate  \cite{golub1965calculating} 
\begin{align} \label{tsvd:eq}
  \hat{X}_r = \sum_{i=1}^r y_i \V{u}_i \V{v}_i'
\end{align}
where $r$ is an estimate of $rank(X)$, whose choice in practice tends to be ad hoc \cite{Gavish2014}.

Recently, \cite{Shabalin2013,candes2013unbiased,gavish2017optimal} have shown that under white
additive noise,
it is useful to apply a carefully
designed {\em shrinkage} function $\eta:\R\to\R$ to the data singular values,
and proposed estimators of the form
\begin{align} \label{shrinkage:eq}
  \hat{X}_\eta =  \sum_{i=1}^n \eta(y_i) \V{u}_i\V{v}_i'\,.
\end{align}
Such estimators are extremely simple to use, as they involve only simple
manipulation of the data singular values. Interestingly, in the additive white noise
case, it was shown that a unique optimal shrinkage function $\eta(y)$ exists, 
which asymptotically delivers the same performance 
as the best possible rotation-invariant estimator based on the data $Y$ 
\cite{gavish2017optimal}. 
Singular value shrinkage thus emerged as a simple yet
highly effective method for improving the SVD in the presence of 
white additive noise, with the unique optimal shrinker as a natural choice for
the shrinkage function. A typical form of optimal singular value shrinker is
shown in Figure \ref{fig:shrinker} below, left panel.

Shrinkage of singular values, an idea that can be traced back to
Stein's groundbreaking work on covariance estimation from the 1970's \cite{Stein1986}, 
is a natural generalization of the classical TSVD. Indeed, $\hat{X}_r$ 
is equivalent to shrinkage with the {\em hard thresholding} shrinker
$\eta(y)=\mathbf{1}_{y\geq \lambda}$, as \eqref{tsvd:eq} is
equivalent to  
\begin{align} \label{hard:eq}
  \hat{X}_\lambda = \sum_{i=1}^n \mathbf{1}_{y_i\geq \lambda}
  \V{u}_i \V{v}_i'\,
\end{align}
with a specific choice of the so-called {\em hard threshold} $\lambda$.
While the choice of the rank $r$ for truncation point 
TSVD is often ad hoc and based on
gut feeling methods such as the Scree Plot method \cite{Cattell1966}, its
equivalent formulation, namely hard thresholding of singular values, allows formal and systematic
analysis.
In fact, restricting attention to hard thresholds alone  \cite{Gavish2014} 
has shown that under white additive noise there exists a unique asymptotically
optimal choice of hard threshold for singular values. The optimal hard threshold
is a systematic, rational choice for the number of singular values that should
be included in a truncated SVD of noisy data. 
\cite{Nadakuditi2014}  has proposed an algorithm that finds $\eta^*$ in 
presence of additive noise and missing values, but has not derived an 
explicit shrinker.

\subsection{Overview of main results}
In this paper, we extend this analysis to common data contaminations 
that go well beyond additive white noise,
including an arbitrary combination of additive
noise, multiplicative noise, missing-at-random entries, uniformly distributed
outliers and uniformly distributed corrupt entries. 

The primary contribution of this paper is formal proof that there exists a
unique asymptotically
optimal shrinker for singular values under 
uniformly random data contaminations,
as well a unique asymptotically optimal hard threshold. Our results are based on
a novel, asymptotically precise description of the effect of these data
contaminations on the singular values and the singular vectors of the data
matrix, extending the technical contribution of 
\cite{Shabalin2013,Nadakuditi2014,gavish2017optimal} to the  setting of 
general uniform data contamination.

\paragraph{General contamination model.}
We introduce the model 
\begin{eqnarray} \label{model:eq}
Y = A\odot X+ B
\end{eqnarray}
where $X$ is the target matrix to be recovered, and $A,B$ are random
matrices with i.i.d entries. 
Here, $(A\odot B)_{i,j} = A_{i,j}B_{i,j}$ is the Hadamard (entrywise) product
of $A$ and $B$.

Assume that $A_{i,j}\iid (\mu_A,\sigma_A^2)$, meaning that the entries of $A$ are
i.i.d drawn from a distribution with mean $\mu_A$ and variance $\sigma^2_A$, and
that $B_{i,j}\iid(0,\sigma_B^2)$.
In Section \ref{sec:modes} we show that for various choices of
the matrix $A$ and $B$, 
this model represents a broad range of uniformly distributed
random contaminations,
including an arbitrary combination of additive
noise, multiplicative noise, missing-at-random entries, uniformly distributed
outliers and uniformly distributed corrupt entries.
As a simple example, if $B\equiv 0$ and $P(A_{i,j}=1) = \kappa$, then
the $Y$ simply has missing-at-random entries.

To quantify what makes a ``good'' singular value shrinker 
$\eta$ for use in
\eqref{shrinkage:eq}, we use the standard  Mean Square Error (MSE) metric and 
\[ 
  L(\eta|X) = \norm{\hat{X}_\eta(Y)-X}_F^2\,.  
\] 
Using the methods of \cite{gavish2017optimal},
our results can easily be extended to
other error metrics, such as the nuclear norm or operator norm losses.  
Roughly speaking, an optimal shrinker $\eta^*$ has the property that,
asymptotically as the matrix size grows,
\[ \label{optimal_shrinker:def}
  L(\eta^*|X)  \leq L(\eta|X)
\]
for any other shrinker $\eta$ and any low-rank target matrix $X$.

The design of optimal shrinkers requires a subtle understanding of the random fluctuations of
the data singular values $y_1,\ldots,y_n$, which are caused by the random
contamination. Such results in random matrix theory are generally hard to prove,
 as there are nontrivial correlations between $y_i$
and $y_j$, $i\neq j$.  Fortunately, in most applications it is very reasonable to
assume that the target matrix $X$ is low rank.  This allows us to overcome this
difficulty by following \cite{Nadakuditi2014,Shabalin2013,Gavish2014} and 
considering an
asymptotic model for low-rank $X$, inspired by Johnstone's Spiked Covariance
Model \cite{Johnstone2001}, in which the correlation between $y_i$ and $y_j$,
 for $i\neq j$ vanish asymptotically.

 We state our main results informally at first. 
 The first main result of this paper
is the existence of a unique asymptotically optimal hard threshold $\lambda^*$
in \eqref{hard:eq}. 

Importantly, as $\E(Y)=\mu_A X$, to apply hard thresholding to $Y=A\odot X+B$
we must from now on define 
{\small
\[\hat{X}_\lambda = \frac{1}{\mu_A} \sum_{i=1}^n \mathbf{1}_{y_i>\lambda}
\V{u}_i\V{v}_i' \,.\] \label{eq:svht_mu}
}

\begin{thm} \label{thm:informal_svht} 
  (Informal.) 
  Let $X$ be an $m$-by-$n$ low-rank matrix and assume that we observe the
  contaminated data matrix $Y$ given by the general contamination model
  \eqref{model:eq}.
  Then there exists a unique optimal (def. \ref{def:optimal_shrinker}) hard threshold $\lambda^*$ for the singular values of $Y$,
  given by     
  {\small
  \begin{eqnarray*}
      \lambda^* = \sigma_B\sqrt{ \left(c+\frac{1}{c}\right )\left(c+\frac{\beta }{c}\right)}    
    \end{eqnarray*}
  }
      where  $\beta=m/n$ and {\small $c=\sqrt{
      1+\beta+\sqrt{1+14\beta+\beta^2}}/\sqrt{2}$. }  
\end{thm} 
Our second main result is the
existence of a unique asymptotically optimal shrinkage function $\eta^*$ in
(equation \eqref{shrinkage:eq}). 
We
calculate this shrinker explicitly:
\begin{thm}\label{thm:informal_shrinker}
  (Informal.) 
  Assume everything as in Theorem \ref{thm:informal_svht}.
  Then there exists a unique optimal (def. \ref{def:optimal_shrinker}) shrinker $\eta^*$ for the singular values of $Y$ 
 given by
 {\small
  \[
  \eta^*(y) = \begin{dcases} 
    \frac{\sigma_B^2}{y\mu_A}\sqrt{\left( \left(\frac{y}{\sigma_B}\right)^2-\beta-1\right)^2-4\beta}
    & y \geq \sigma_B(1+\sqrt{\beta}) \\
    0  &  y < \sigma_B(1+\sqrt{\beta}) \\
  \end{dcases}
\]
}
\end{thm}

We also discover that for each contamination model, there is a critical
signal-to-noise cutoff, below which $X$ cannot be reconstructed from the
singular values and vectors of $Y$. Specifically, 
let $\eta_0$ be the zero singular value shrinker, $\eta_0(y)\equiv 0$,
  so that $\hat{X}_{\eta_0}(Y) \equiv 0$. 
  Define the {\em critical signal level} for a shrinker $\eta$
  by 
  \[
    x^{critical}(\eta)= \inf_{x} \left\{x: L(\eta|X) < L(\eta_0|X) \right\}
  \]where $X=x\tilde{\V{u}}\tilde{\V{v}}'$ is an arbitrary rank-$1$ matrix with
singular value $x$. 
In other words, $x^{critical}(\eta)$ is the smallest
singular value of the target matrix, for which $\eta$ still
outperforms the
trivial zero shrinker $\eta_0$. 
As we show in Section \ref{sec:results}, a target matrix $X$ with a
singular value  below $x^{critical}(\eta)$ cannot be reliably reconstructed using $\eta$. 
The critical signal level for the optimal shrinker $\eta^*$ 
is
of special importance, since  a target matrix $X$ with a
singular value  below $x^{critical}(\eta^*)$ cannot be reliably reconstructed
using {\em any} shrinker $\eta$. 
Restricting attention to hard thresholds only, 
we define   $x^{critical}(\lambda)$, the critical level for
a hard threshold, similarly. Again, singular values of $X$ that fall below 
$x^{critical}(\lambda^*)$ cannot be reliably reconstructed using {\em any} hard
threshold.

Our third main result is the explicit calculation of these critical
signal levels:
\begin{thm} \label{thm:informal_critical}(Informal.) Assume everything 
  as in Theorem \ref{thm:informal_svht} and let $c$ be as in  Theorem
  \ref{thm:informal_svht}.
  Let $\eta^*$ be the optimal shrinker from Theorem \ref{thm:informal_shrinker} and let
  $\lambda^*$ be the optimal hard threshold from Theorem \ref{thm:informal_svht}.
  The critical signal levels for $\eta^*$ and $\lambda^*$ are  given by:  
  \begin{eqnarray*}
  x^{critical}(\eta^*)&=& (\sigma_B /\mu_A) \cdot \beta^\frac{1}{4} \\
  x^{critical}(\lambda^*)&=& (\sigma_B / \mu_A) \cdot c \,.
\end{eqnarray*}
  \end{thm}

Finally, one might ask what the improvement is in terms of the mean square error 
that is   guaranteed  by using the optimal shrinker and optimal threshold. 
As discussed below, existing methods are either infeasible in terms of running
time on medium and large matrices, or lack a theory that can predict the
reconstruction mean square
error. For lack of a better candidate, we compare the optimal shrinker and
optimal threshold to the default method, namely, TSVD. 

\begin{thm} \label{improve:thm} (Informal.)
  Consider $\beta=1$, and denote the worst-case mean square error of TSVD, $\eta^*$
  and $\lambda^*$ by $M_{TSVD}$, $M_{\eta^*}$ and $M_{\lambda^*}$, respectively,
  over a target matrix of low rank $r$.
  Then 
  {\small
    \begin{eqnarray*}
      M_{TSVD}&=& \left(\frac{\sigma_B}{\mu_A} \right)^2 5r\\
      M_{\eta^*} &=& \left(\frac{\sigma_B}{\mu_A} \right)^2 2r\\
      M_{\lambda^*} &=&  \left(\frac{\sigma_B}{\mu_A} \right)^2 3r  \,.
    \end{eqnarray*}
  }
\end{thm}
Indeed, the optimal shrinker offers a significant performance improvement
(specifically, an improvement of $3r(\sigma_B/\mu_A)^2$,  over
  the TSVD baseline.

\begin{figure}[h]
  \centering
  \includegraphics[width=.45\linewidth]{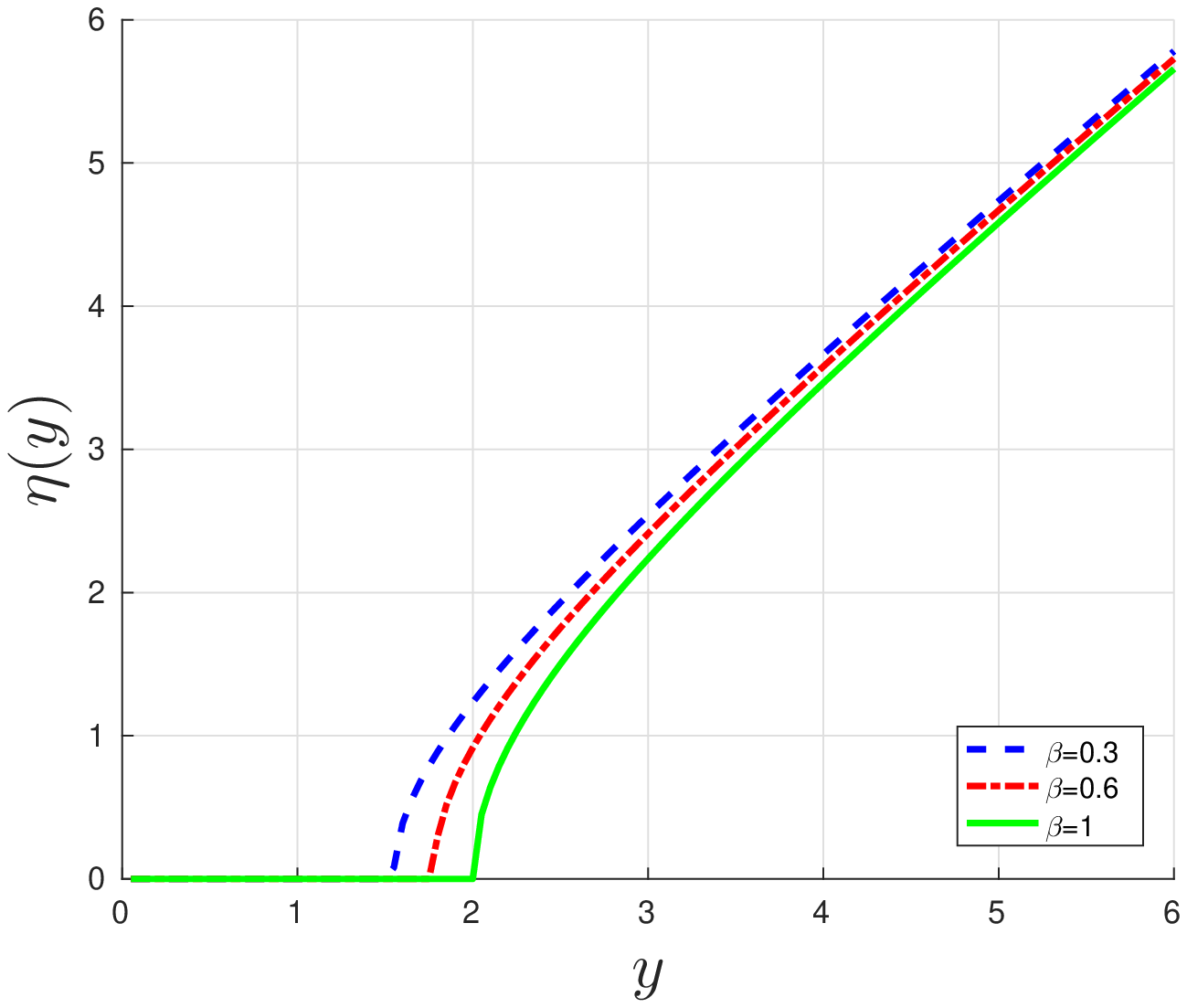}
  \includegraphics[width=.45\linewidth]{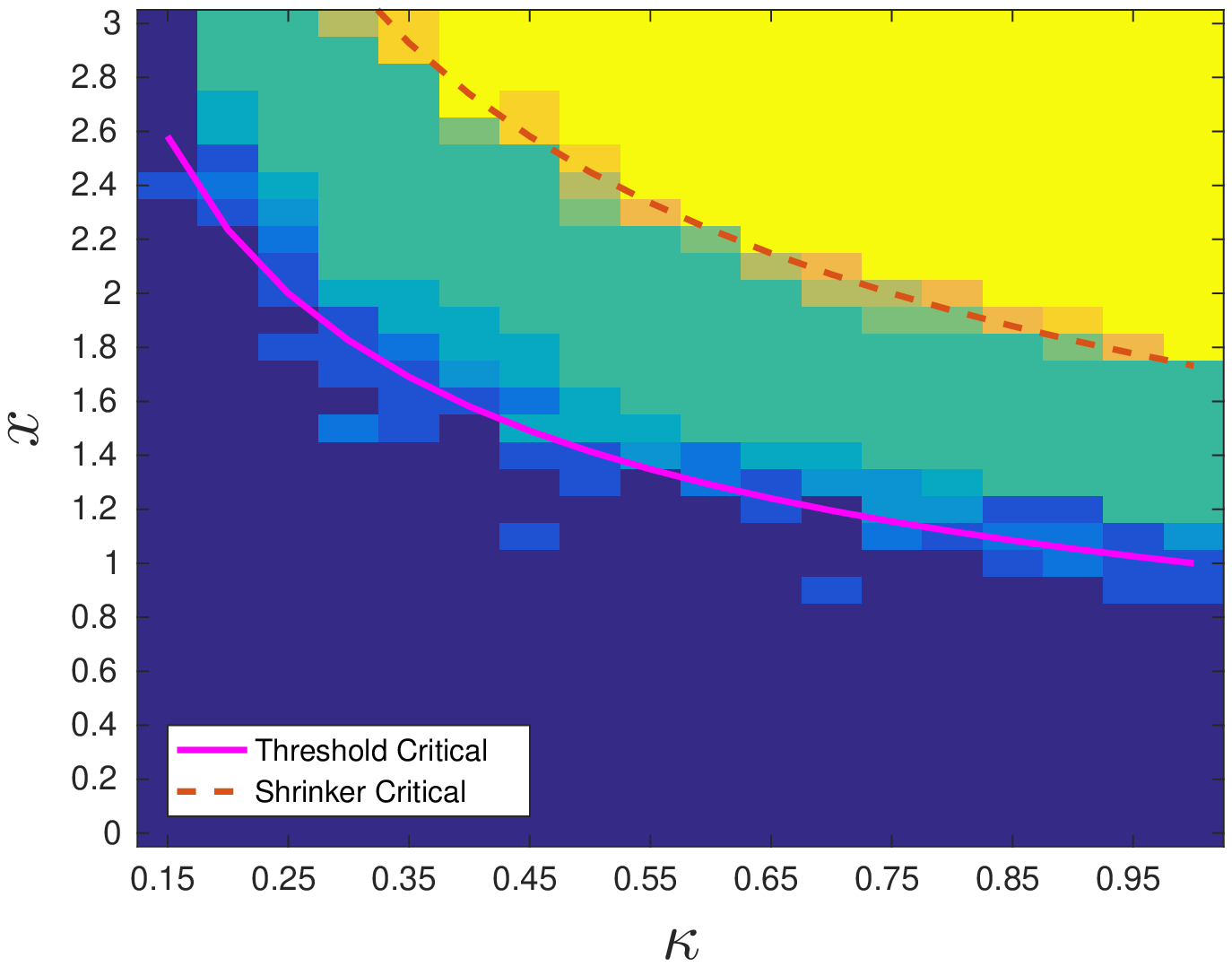}
  \caption{Left: Optimal shrinker for additive noise and missing-at-random
    contamination. Right: Phase plane for critical signal levels, see Section
  \ref{sec:verifying_results}, Simulation 2.}
  \label{fig:shrinker}
\end{figure}

Our main results allow easy calculation of the optimal threshold, optimal
shrinkage and signal-to-noise cutoffs 
for various specific contamination models. For example:

\begin{enumerate} 
   \item {\bf Additive noise and missing-at-random.} 
     Let $X$ be an $m$-by-$n$ low-rank matrix. Assume that some entries are
     completely missing and the rest suffer white additive noise. 
     Formally, we observe the contaminated matrix 
    \[ 
        Y_{i,j} = \begin{cases} X_{i,j} + Z_{i,j}&
	\text{w.p. }\kappa\\ 0  & \text{w.p. }1-\kappa\end{cases}\,,
      \]  
      where $Z_{i,j}\iid (0,\sigma^2)$, namely, follows an unknown distribution with mean
      $0$ and variance $\sigma^2$.
      Let $\beta = m/n$. 
      Theorem \ref{thm:informal_svht} implies that in this case, 
      the optimal hard threshold for the singular values of $Y$ is 
    \[
      \lambda^* =\sqrt{\sigma^2\kappa 
      \left(c+1/c\right )\left(c+\beta/c\right)} 
        \]
	where
        {\small  $c = \sqrt{ 1+\beta+\sqrt{1+14\beta+\beta^2}}/\sqrt{2}$}.
In other words, the optimal location (w.r.t mean square error) to truncate the singular values of $Y$, in
order to recover $X$, is given by $\lambda^*$.
The optimal shrinker from Theorem \ref{thm:informal_shrinker} 
for this contamination mode may be calculated similarly,
and is shown in Figure \ref{fig:shrinker}, left panel. 
By Theorem \ref{improve:thm}, 
the improvement in mean square error obtained by using the optimal shrinker, 
over the TSVD baseline, is $3r\sigma^2/\kappa$, quite a
significant improvement.
  \item{\bf Additive noise and corrupt-at-random.} 
     Let $X$ be an $m$-by-$n$ low-rank matrix. Assume that some entries are
     irrecoverably corrupt (replaced by random entries),
     and the rest suffer white additive noise. 
     Formally,
    \[ 
      Y_{i,j} = \begin{cases} X_{i,j} + Z_{i,j}&
	\text{w.p. }\kappa\\ W_{i,j}  & \text{w.p. }1-\kappa\ \end{cases}\,.  
    \]  
    Where $Z_{i,j}\iid (0,\sigma^2)$, 
    $W_{i,j}\iid (0,\tau^2)$, and $\tau$ is typically large.
    Let $\tilde{\sigma} = \sqrt{\kappa\sigma^2 +(1-\kappa)\tau^2}$. 
    The optimal shrinker, which should be applied to the 
    singular values of $Y$, is given by:
    {\small
\[
  \eta^*(y) = \begin{dcases} 
    \tilde{\sigma}^2/(y\kappa)\sqrt{\left( 
      \left(y/\tilde{\sigma}\right)^2-\beta-1\right)^2-4}  & y \geq \tilde{\sigma}(1+\sqrt{\beta}) \\
    0  &  y < \tilde{\sigma}(1+\sqrt{\beta}) \\
  \end{dcases}
  \,.
\]
}
By Theorem \ref{improve:thm}, 
the improvement in mean square error, obtained by using the optimal shrinker, 
over the TSVD baseline, is $3r(\kappa\sigma^2+(1-\kappa)\tau^2)/\kappa^2$.  %

\end{enumerate}

\subsection{Related Work}
The general data contamination model we propose includes as special cases
several modes extensively studied in the literature, including
missing-at-random and outliers. While it is impossible to propose a complete list
of algorithms to handle such data, we offer a few pointers, organized around the notions of
robust principal component analysis (PCA) and matrix completion. 
To the best of our knowledge, the precise effect of general data contamination on the
SVD (or the closely related PCA) has not been documented thus far. 
The approach we propose, based on careful manipulation of the data singular
values, enjoys three distinct advantages. One, its running time is not
prohibitive; indeed, it involves a small yet important modification on top of the
SVD or TSVD, so that it is available whenever the SVD is available.
Two, it is well understood and its performance (say, in mean square
error) can be reliably predicted  by
the available theory. Three, to the best of our knowledge, 
none of the approaches below have become mainstream,
and most practitioners still turn to the SVD, even in the presence of data
contamination. Our approach 
can easily be used in practice, as it relies on the well-known and very widely used SVD, and can be
implemented as a simple modification on top of the existing SVD implementations.

{\bf Robust Principle Component Analysis (RPCA).} 
    In RPCA, one assumes $Y=X+W$ where $X$ is the low rank target
    matrix and $W$ is a sparse outliers matrix. 
    Classical approaches such as influence functions \cite{huber2011robust}, multivariate
    trimming \cite{gnanadesikan1972robust} and random sampling techniques
    \cite{fischler1981random} lack a formal theoretical framework and are not
    well understood. 
       More modern approaches based on convex optimization 
    \cite{Wright2009,Candes2011} proposed reconstructing $X$ from $Y$ via the nuclear norm
    minimization 
    \[ \underset{X}{min} \norm{X}_*+\lambda\norm{Y-X}_1\,, \]
    whose runtime and memory requirements are both prohibitively large in medium and large matrices.

    {\bf Matrix Completion.} 
    There are numerous heuristic approaches for data analysis in the presence of missing values \cite{buuren2011mice,schafer1997analysis,rubin1996multiple}.
  To the best of our knowledge, there are no formal guarantees of their
  performance. 
  When the target matrix is known to be low rank, the reconstruction problem is
  known as matrix completion.  \cite{Candes2010,Candes2011,candesPlan2010matrix}
  and numerous other authors have shown that a semi-definite program may be used
  to stably recover the target matrix, even in the presence of additive noise. 
  Here too, the runtime and memory requirements are both prohibitively large in
  medium and large matrices, making these algorithms infeasible in practice.

    \section{A Unified Model for Uniformly Distributed Contamination} \label{sec:modes}

Contamination modes encountered in practice are best described by a
combination of primitive modes, shown in Table 
\ref{tab:primitive} below. 
These primitive contamination modes fit into a single template:
\begin{defn} \label{f:def}
  {\em
  Let $A$ and $B$ be two random variables, and assume that all moments of
$A$ and $B$ are bounded. Define the {\em contamination link
function}
\[
  f_{A,B}(x)= Ax+B\,.
\]
Given a matrix $X$, define 
the corresponding contaminated matrix $Y$ with entries
\begin{eqnarray} \label{eqn:framework} Y_{i,j} \indep f_{A,B}(X_{i,j}) \,.
\end{eqnarray} 
}
\end{defn}

Now observe that each of the primitive modes above corresponds to a different choice
of random variables $A$ and $B$, as shown in Table \ref{tab:primitive}.
Specifically, each of the primitive modes is described by a different assignment
to $A$ and $B$. We employ three different random variables in these assignments:
$Z\iid(0,\sigma^2/n)$, a random variable describing multiplicative or additive
noise; $W\iid(0,\tau^2/n)$, a random variable describing a large ``outlier''
measurement; and $M\iid Bernoulli(\kappa)$ describing a random choice of
``defective'' entries, such as a missing value, an outlier and so on.

\begin{table*}[h!] \centering {\renewcommand{\arraystretch}{1.2}\caption{
      \em \small Primitive modes fit into the model \eqref{eqn:framework}.  By
      convention, $Y$ is $m$-by-$n$,
      $Z\iid (0,\sigma^2/n)$ denotes a noise random variable, 
      $W\iid (0,\tau^2/n)$ denotes an outlier
  random variable and $M\iid Bernoulli(\kappa)$ is a contaminated entry
random variable. }
\label{tab:primitive}
  \begin{tabular}{|l|l|c|c|c|}\hline mode & model & A & B & levels \\ 
    \hline 
    i.i.d additive noise & $Y_{i,j}=X_{i,j}+Z_{i,j}$ & $1$ & $Z$ 
    & $\sigma$ \\ 
    \hline 
    i.i.d multiplicative noise & $Y_{i,j}=X_{i,j}\,Z_{i,j}$ & $Z$ & $0$ 
    & $\sigma$ \\
    \hline
    missing-at-random & $Y_{i,j} = M_{i,j} \, X_{i,j}$ & $M$ & $0$ 
    & $\kappa$ \\
    \hline
    outliers-at-random & $Y_{i,j} = X_{i,j} + M_{i,j} W_{i,j}$ & $1$ & $MW$
    & $\kappa$,$\tau$ \\
    \hline 
    corruption-at-random & $Y_{i,j} = M_{i,j}X_{i,j} +
    (1-M_{i,j})W_{i,j}$ & $M$ & $(1-M)W$ 
    & $\kappa$,$\tau$ \\ 
    \hline \end{tabular}}
   \end{table*}

Actual datasets rarely demonstrate a single primitive contamination mode. 
To adequately describe contamination observed in practice,
 one usually needs to
combine two or more of the primitive contamination modes into a composite mode. 
While
there is no point in enumerating all possible combinations,
Table \ref{tab:composite} offers a few notable composite examples, 
using 
the framework \eqref{eqn:framework}.  Many other examples are possible of
course.

\begin{table*} \centering {\renewcommand{\arraystretch}{1.2} 
    \caption{\em \small Some examples of composite
  contamination modes and how they fit into the model \eqref{eqn:framework}.
$Z$,$W$,$M$ are the same as in Table \ref{tab:primitive}.}
\label{tab:composite}
  \begin{tabular}{|c|c|c|c|}\hline mode  & A & B & levels\\ \hline Additive
      noise and
    missing-at-random & 
    $M$ &  $ZM$ & $\sigma$,$\kappa$ \\ \hline Additive noise and corrupt-at-random & 
    $M$ & $ZM + W(1-M)$ & $\sigma$,$\kappa$,$\tau$\\ \hline multiplicative noise  and  corrupt-at-random & 
  $ZM$ & $W(1-M)$ & $\sigma$,$\kappa$,$\tau$ \\ \hline Additive noise
    and outliers& 
  $1$ & $Z+W(1-M)$ & $\sigma$,$\kappa$,$\tau$ \\ \hline \end{tabular}}
\end{table*}

\section{Signal Model} \label{sec:signal} 

Following \cite{Shabalin2013} and \cite{Gavish2014}, 
as we move toward our formal results 
we are considering an
asymptotic model inspired by Johnstone's Spiked Model \cite{Johnstone2001}.
Specifically, we are considering a 
sequence of increasingly larger data target matrices $X_n$, and 
corresponding data matrices 
$Y_n \iid f_{A_n,B_n}(X_n)$. 
We make the following assumptions regarding the matrix sequence $\{X_n\}$:

\begin{enumerate}
  \item[\bf A1] {\em Limiting aspect ratio:}  The matrix dimension $m_n \times n$ sequence converges:  $m_n/n\to \beta$ as $n\to\infty$. To simplify the results, we assume $0< \beta\leq 1$.
      \item[\bf A2] {\em Fixed signal column span:} 
        Let the rank $r>0$ be fixed and choose a vector $\V{x}\in\R^r$ with
        coordinates $\V{x}=(x_1,\ldots x_r)$ such that $x_1>\ldots >x_r>0$.
        Assume that for all $n$ 
        \[
          X_n = \tilde{U}_n \, diag(x_1,\ldots,x_r) \tilde{V}_n
        \]
        is an arbitrary singular value decomposition of $X_n$, 
	%$\tilde{U}_n$ and $\tilde{V}_n$ 
      \item[\bf A3] {\em Incoherence of the singular vectors of $X_n$:} We make one of the
        following two assumptions regarding the singular vectors of $X_n$:
        \begin{itemize}
          \item[\bf A3.1] $X_n$ is random with an orthogonally invariant distribution.
              Specifically, $\tilde{U}_n$ and $\tilde{V}_n$, which follow the Haar distribution
            on orthogonal matrices of size $m_n$ and $n$, respectively. 
          \item[\bf A3.2] The singular vectors of $X_n$ are non-concentrated. Specifically, 
         each left singular vector $\Vt{u}_{n,i}$ of $X_n$ (the $i$-th column of
         $\tilde{U}_n$) and each right singular vector $\Vt{v}_{n,j}$ of $X_n$ (the
         $j$-th column of $\tilde{V}_n$) satisfy\footnote{The incoherence assumption is widely used in related literature
        \cite{Chand2011,Nadakuditi2014,CandesSVD2010}, and asserts that the singular
        vectors are spread out so $X$ is not sparse and does not share singular
      subspaces with the noise.}
         \begin{eqnarray*}
             \norm{\Vt{u}_{n,i}}_\infty \leq C \frac{\log^D(m_n)}{\sqrt{m_n}}
             \qquad\text{and} \qquad
             \norm{\Vt{v}_{n,j}}_\infty \leq C \frac{\log^D(n)}{\sqrt{n}}
         \end{eqnarray*}
         for any $i,j$ and fixed constants $C,D$.
        \end{itemize}
	    %\item {\em Bounded moments for the contamination variables.} All moments of the random
    \end{enumerate}

    \begin{defn} \label{signal:def} {\bf (Signal model.)}
  Let $A_n\iid (\mu_A,\sigma_A^2/n)$ 
  and $B_n\iid (0,\sigma_B^2/n)$ have bounded moments.
  Let $X_n$ follow assumptions {\bf [A1]}--{\bf[A3]} 
  above. We say that 
   the matrix sequence
   $Y_n=f_{A_n,B_n}(X_n)$ follows our signal model, where
   $f_{A,B}(X)$ is as in Definition \ref{f:def}. 
  We further denote 
  $X_n = \sum_{i=1}^r x_i \Vut_{n,i} \Vvt_{n,i}$ for the singular value
  decomposition of $X_n$ and 
  $Y_n = \sum_{i=1}^m y_{n,i} \Vu_{n,i} \Vv_{n,i}$
  for the singular value decomposition of $Y_n$.
\end{defn}

\section{Main Results} \label{sec:results}
Having described the contamination and the signal model, we can now formulate our
main results.  All proofs are deferred to the Supporting Information. Let $X_n$ and $Y_n$ follow our signal model, Definition
\ref{signal:def}, and write $\V{x}=(x_1,\ldots,x_r)$ for the non-zero singular
values of $X_n$. 
For a shrinker $\eta$,   we write
\[
L_\infty(\eta|\V{x}) \aseq  \lim_{n\to\infty}\norm{\hat{X}_n(Y_n)-X_n}_F^2.
  \]
  assuming the limit exists almost surely. The special case of hard
  thresholding at $\lambda$ is denoted as   $L_\infty(\eta|\V{x})$. 

  \begin{defn} \label{def:optimal_shrinker}
    {\bf Optimal shrinker and optimal threshold.} 
    A shrinker $\eta^*$ is called {\em optimal} 
    if 
    \[
      L_\infty(\eta|\V{x}) \leq L_\infty(\eta|\V{x}) 
    \]
    for any shrinker $\eta$, any 
    $r\geq 1$ and any $\V{x}=(x_1,\ldots,x_r)$. Similarly, 
    a threshold $\lambda$ is
    called optimal if $ L_\infty(\lambda^*|\V{x}) \leq L_\infty(\lambda|\V{x}) $
for any threshold $\lambda$, any 
    $r\geq 1$ and any $\V{x}=(x_1,\ldots,x_r)$.
  \end{defn}

  With these definitions, our main results Theorem \ref{thm:informal_shrinker} and
  Theorem \ref{thm:informal_svht} become formal. To make 
  Theorem \ref{thm:informal_critical} formal, we need the following lemma and definition.

  \begin{lemma} \label{lemma:decomp}
    {\bf Decomposition of the asymptotic mean square error.}
    Let $X_n$ and $Y_n$ follow our signal model (Definition
    \ref{signal:def}) and write $\V{x}=(x_1,\ldots,x_r)$ for the non-zero singular
    values of $X_n$, and let $\eta$ be the optimal shrinker.  
    Then
 the limit $L_\infty(\eta|\V{x})$ a.s. exists, and
    {\small
    $
       L_\infty(\eta|\V{x}) \aseq \sum_{i=1}^r  L_1(\eta|x)
     $},
where
{\small
\begin{align*}
	 L_1(\eta|x)   = 
	\begin{dcases*}
		x^2  \left( 1 - \frac {(t^4-\beta)^2}{(t^4+\beta t^2)(t^4+t^2)}  \right)
		& $t\geq \beta^{\frac{1}{4}}$\\
		x^2 & $t<  \beta^{\frac{1}{4}}$\\
	\end{dcases*}
\end{align*}
}
where $t = (\mu_A\cdot x)/\sigma_B$.
    Similarly, for a threshold $\lambda$ we have
    {\small
    $
      L_\infty(\lambda|\V{x}) = \sum_{i=1}^r  L_1(\lambda|x)
    $
  }
   with
    {\small
\begin{align*}
	L_1(\lambda|x)   = 
	\begin{dcases*}
	  \left(\frac{\sigma_B}{\mu_A} \right) ^2 \left( \left( t+\frac{1}{t}\right) \left(t +\frac{\beta }{t}\right)- 
                \left(t^2 - \frac{2\beta  }{t^2} \right)\right)
&$ \mu_A x\geq x(\lambda) $\\
	x^2 & $\mu_A x< x(\lambda)$\\
	\end{dcases*}
\end{align*}
}

Where
{\small
  \begin{eqnarray}\label{eq:x_of_y}
  x(y)   =
	\begin{dcases*}
          (\sigma_B/\sqrt{2}\mu_A) \sqrt{\left(y/\sigma_B\right)^2
                  -\beta-1+
                \sqrt{\left(1+\beta-\left(y/\sigma_B\right)^2\right)^2-4\beta}}
		  &$t\geq \beta^{\frac{1}{4}}$\\
	0 & $ t < \beta^{\frac{1}{4}} $\\
	\end{dcases*}
\end{eqnarray}
}

\end{lemma}

\begin{defn}
  Let $\eta_0$ be the zero singular value shrinker, $\eta_0(y)\equiv 0$,
  so that $\hat{X}_{\eta_0}(Y) \equiv 0$. 
  Let $\eta$ be a singular value shrinker. The critical signal level for $\eta$
  is 
  \[
    x^{critical}(\eta)= \inf_{x} \left\{ L_1(\eta|X) < L_1(\eta_0|X) \right\}
\]
\end{defn}

As we can see, the asymptotic mean square error decomposes over the singular
values of the target matrix, $x_1,\ldots,x_r$. Each value $x_i$ that falls below 
$x^{critical}(\eta)$ is better estimated with the zero shrinker $\eta_0$ than
with $\eta$. It follows that any $x_i$ that falls below  $x^{critical}(\eta^*)$, where $\eta^*$ is the
optimal shrinker, cannot be reliably estimated by any shrinker $\eta$, and its
corresponding data singular value $y_i$ should simply be set to zero.  
This makes Theorem \ref{thm:informal_shrinker} formal.

\section{Estimating the model parameters}

In practice, using the optimal shrinker we propose requires 
an estimate of the model parameters.
In general, $\sigma_B$ is easy to estimate from the data
via a median-matching method \cite{Gavish2014}, namely
\begin{align*}
  \hat{\sigma}_B = \frac{y_{med}}{\sqrt{n \mu_\beta}}\,,
\end{align*}
where $y_{med}$ is the median singular value of Y, and $\mu_\beta$ is the median
of the Mar\u{c}enko-Pastur distribution.
However, estimation of $\mu_A$ and $\sigma_A$ 
must be considered on a case-by-case basis. 
For example, in the ``Additive noise and missing at random'' mode 
(table \ref{tab:composite}), 
$\sigma_A\equiv 1$ is known,
and $\mu_A$ is estimated by dividing the amount of missing values by the matrix size.

  \section{Simulation}\label{sec:verifying_results}
  Simulations were performed to verify the correctness of our main results\footnote{The full Matlab code that generated the figures in this paper and in the
  Supporting Information is permanently available at
\url{https://purl.stanford.edu/kp113fq0838}.}. 
  For more details, see Supporting Information.
  
  \begin{enumerate} 
    \item {\bf Critical signal level $x^{critical}(\lambda^*)$
      under increasing noise.} Figure \ref{fig:emp1}, left panel, shows the amount of
      data singular values $y_i$ above $x^{critical}(\lambda^*)$, as a function
      of the fraction of missing values $\kappa$. Theorem
      \ref{thm:informal_critical} correctly predicts the exact values of
      $\kappa$ at which the ``next'' data singular value falls below
      $x^{critical}(\lambda^*)$. 

    \item {\bf Phase plane for critical signal levels 
      $x^{critical}(\eta^*)$ and $x^{critical}(\lambda^*)$.}
      Figure \ref{fig:shrinker}, right panel, shows 
  the $x,\kappa$  plane, where $x$ is the signal level and $\kappa$ is the
  fraction of missing values. At each point in the plane, several independent
  data matrices were generated. Heatmap shows the fraction of the experiments at which the data singular value 
  $y_1$ was above  $x^{critical}(\eta^*)$ and $x^{critical}(\lambda^*)$.
  The overlaid graphs are theoretical predictions of the critical points.

    \item {\bf Brute-force verification of the optimal shrinker shape.}
      Figure \ref{fig:emp1}, right panel, shows the shape of the optimal shrinker (Theorem \ref{thm:informal_svht}).
      We performed a brute-force search for the value of $\eta(y)$
      that produces the minimal mean square error. A brute force search, performed
      with a
      relatively small matrix size, matches the asymptotic shape of the optimal
      shrinker.
  \end{enumerate}

  \begin{figure}[h]\centering
  \includegraphics[width=.45\linewidth,valign=t]{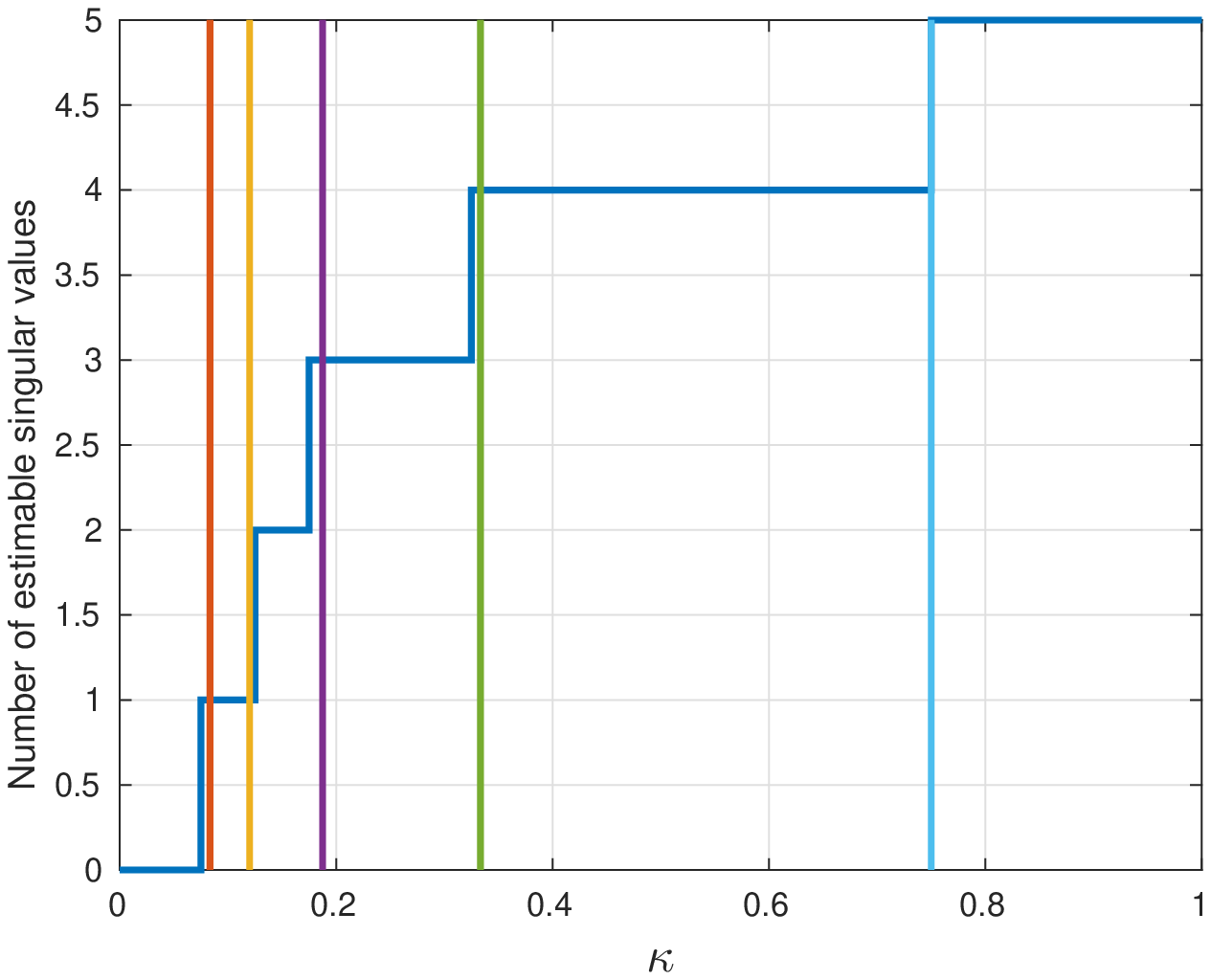}
 \includegraphics[width=.45\linewidth,valign=t]{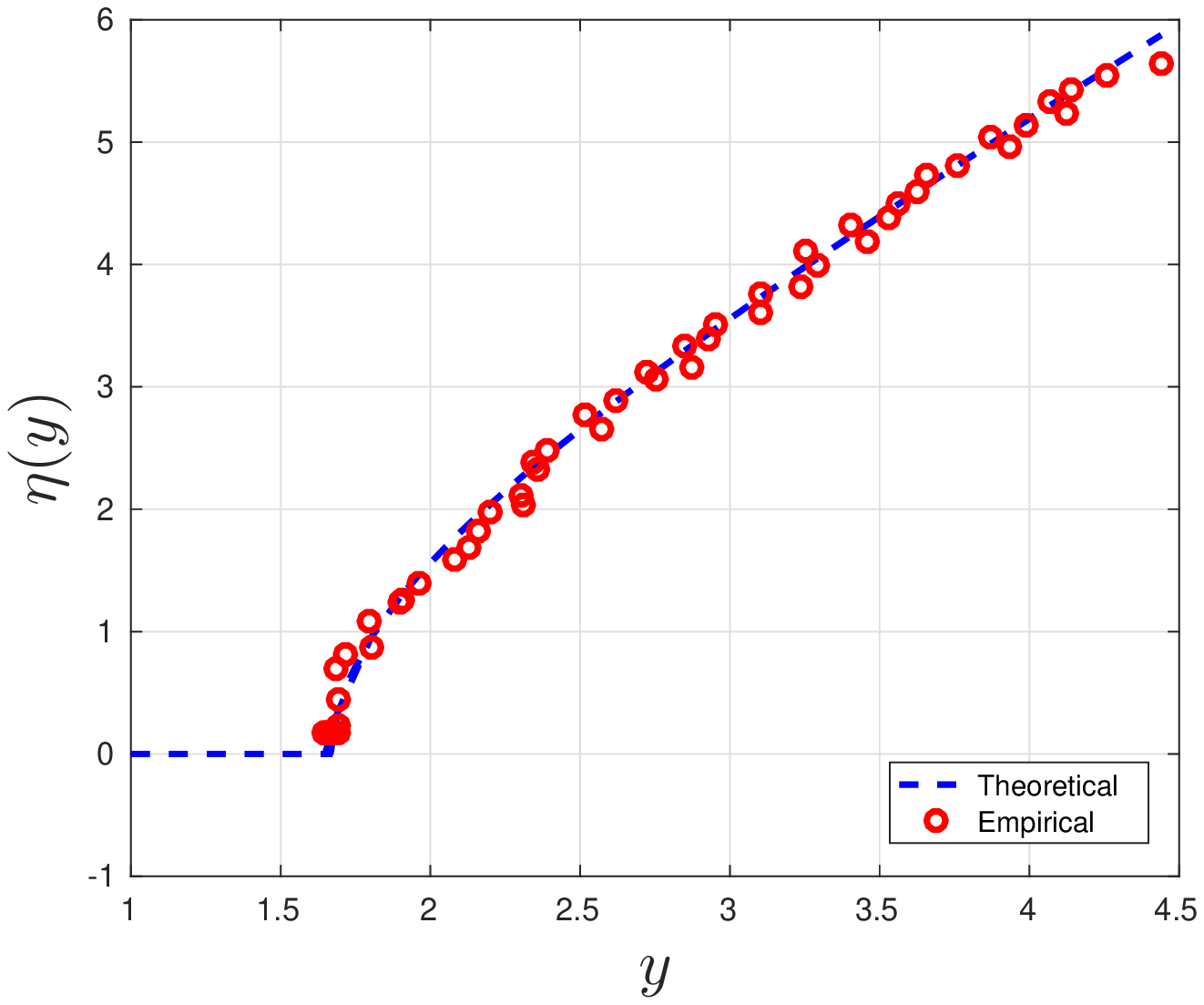} 
 \caption{Left: empirical validation of the predicted 
 critical signal level (Simulation 1).
 Right: Empirical validation of the optimal shrinker shape (Simulation 3).}
\label{fig:emp1}

\end{figure}

\ifx\SI\undefined\else

\section{Simulation details}

  \begin{enumerate}
    \item {\bf 
      Critical signal level $x^{critical}(\lambda^*)$ under increasing noise.}
      Consider the composite mode ``i.i.d additive noise and missing-at-random''(sec. \ref{sec:modes}, mode \ref{sec:modes})
    with noise level $\sigma=1$, singular values $\V{x}=(2,3,4,5,6)$ and matrix size $m=n=1000$. 
    To test whether our main results correctly predict the values of $\kappa$ at
    which the signal singular values stop being estimable, we scanned the $\kappa$
    axis. At each value of $\kappa$ we counted the number of estimable singular
    values and compared with the cutoff values of $\kappa$, at which this number
    should change.
    Figure \ref{fig:emp1} shows the number of measured estimable singular
    values against contamination level $\kappa$, with overlaid cut levels
    predicted by our main results (vertical lines).

    \item {\bf Phase plane for critical signal levels 
      $x^{critical}(\eta^*)$ and $x^{critical}(\lambda^*)$.}
  Figure \ref{fig:shrinker}, right panel, shows 
  the $x,\kappa$  plane, where $x$ is the signal level and $\kappa$ is the
  fraction of missing values. At each point in the plane, several independent
  experiments (data matrices) were generated. The heatmap shows the 
  the fraction of the experiments at which the data singular value 
  $y_1$ was above  $x^{critical}(\eta^*)$ and $x^{critical}(\lambda^*)$.
  Overlaid are the theoretical prediction of the critical points. Matrix size is $600 \times 600$, Monte carlo is 5 and $\beta = 1$.

    \item {\bf Brute-force verification of the optimal shrinker shape.}
      Figure (\ref{fig:emp1},right) shows the shape of the optimal shrinker (Theorem \ref{thm:informal_svht}).
      We performed a brute-force calculation scanning for the value of $\eta(y)$
      that produces the lower mean square error. Brute force search, peformed in
      a $250 \times 250$ matrix, matches the asymptotic shape of the optimal
      shrinker. The scan set the noise level $\sigma$ to 1, missing data level to $1-\mu_A = 0.3$, $\beta = m/n = 1$, and scanned signal values in range [0,6].

  \end{enumerate}

  \section{Proofs}
In this section we prove our main results: Theorem \ref{thm:informal_svht},
Theorem \ref{thm:informal_shrinker}, Theorem \ref{thm:informal_critical} and
Theorem \ref{improve:thm}, as well
as Lemma \ref{lemma:decomp}. 
The proofs rely on the following fundamental lemmas.
\begin{lemma} \label{fundamental:lem}
  Let $Y_n=f_{A_n,B_n}(X_n)$ 
  be a matrix sequence following our signal model (Definition 
  \ref{signal:def}) with $A_n\iid(\mu_A,\sigma^2_A/n)$ and 
  $B_n\iid(0,\sigma^2_B/n)$.
  Write $\bx_i = \mu_A\cdot x_i$.
  Then the following holds:
  \begin{enumerate}
    \item \label{eq:displacement} For each $1\leq i\leq r$ we have 
\[
	\lim_{n \to \infty} y_{n,i} \aseq 
	\begin{dcases*}
	   \sigma_B \sqrt{\left(
            \frac{\bx_i}{\sigma_B}+\frac{\sigma_B}{\bx_i}\right)
            \left(\frac{\bx_i}{\sigma_B}+\frac{\beta
          \sigma_B}{\bx_i}\right) } & $\bx_i>\sigma_B \beta ^\frac{1}{4}$\\
	\lim_{n \to \infty} y_{n,i} \aseq 
	\sigma_B( 1+\sqrt{\beta}) & ${\bx_i}\leq \sigma_B  \beta ^\frac{1}{4}$ \\ 
      \end{dcases*}
    \]

  \item  \label{eq:rotation} Let $1\leq i\leq r$ and $1\leq j \leq m_n$. 
    If $\bx_i>\sigma_B  \beta ^\frac{1}{4}$ and $j\leq r$, we have
\begin{align}
  d\cdot \lim_{n \to \infty} |\langle \Vut_{i},\Vu_{n,j}\rangle|^2   \aseq 
	\begin{dcases*}
	\frac{(\frac{\bx_i}{\sigma})^4-\beta}{(\frac{\bx_i}{\sigma})^4 +
        \beta (\frac{\bx_i}{\sigma})^2}& $\bx_i=\bx_j$\\
	0 & $\bx_i\neq \bx_j$\\
	\end{dcases*}
\end{align}
and
\begin{align*}
  d\cdot \lim_{n \to \infty} |\langle \Vvt_{i},\Vv_{n,j}\rangle|^2   \aseq 
	\begin{dcases*}
	\frac{(\frac{\bx_i}{\sigma})^4-\beta}{(\frac{\bx_i}{\sigma})^4
        +(\frac{\bx_i}{\sigma})^2}& $\bx_i=\bx_j$\\
	0 & $\bx_i\neq \bx_j$\\
	\end{dcases*}
\end{align*}
Otherwise, both quantities converge a.s to $0$ as $n \to \infty$.

\end{enumerate}

\end{lemma}

\paragraph{Proof of Theorem \ref{thm:informal_svht}.}
By lemma \ref{lemma:AMSE} we have that for any singular value hard threshold estimator $\lambda$,
AMSE is given by \eqref{amse_:eq}.
As for the bottom case in \eqref{amse_:eq}, clearly $\overline{x}^2$ is strictly
increasing in $\overline{x}$, and therefor in $x$ as well.
As for the top case in \eqref{amse_:eq}, direct differentiation shows that it is 
strictly decreasing. It follows that the two functions of $\overline{x}$ intersect
at a unique point. Denote their intersection point by $ \overline{x^{est}} = \mu_A x^{est}$:
\begin{eqnarray*}
  \overline{x^{est}} =  \sigma_B \sqrt{ \frac{ 1+\beta+\sqrt{1+14\beta+\beta^2}}{2}}.  
\end{eqnarray*}

It follows that the following hard threshold for Y's singular values achieves the minimum AMSE between the above two expressions:

 \begin{align*}
\lambda^* = \sigma_B \big(\sqrt{( \frac{\overline{x^{est}}}{\sigma_B}+\frac{\sigma_B}{\overline{x^{est}}})(\frac{\overline{x^{est}}}{\sigma_B}+\frac{\beta \sigma_B}{\overline{x^{est}}}})\big) 
\end{align*}.
\qed

\paragraph{Proof of Theorem \ref{thm:informal_shrinker}.}
Recall that for any matrix $X$,  $\norm{X}_F^2=\langle X,X \rangle$ where  
$\langle \cdot ,\cdot\rangle$ is the Hilbert–Schmidt matrix inner product.
The notation $\langle X,X \rangle$ refers to  Hilbert–Schmidt matrix inner
product or to the standard Euclidean inner product.
Denoting our shrinker by $\eta=\eta(y)$,and let $\overline{\eta}$, we evaluate the mean square error:
\begin{align*}
  \norm{\frac{1}{\mu_A}\hat{X}_\eta(Y)-X}_F^2 &=  \langle \frac{1}{\mu_A}\hat{X}_\eta(Y)-X,\frac{1}{\mu_A}\hat{X}_\eta(Y)-X \rangle_{HS}  =\\
  &\langle\frac{1}{\mu_A}\hat{X}_\eta(Y),\frac{1}{\mu_A}\hat{X}_\eta(Y)\rangle + \langle X_n, X_n \rangle -2\langle \frac{1}{\mu_A}\hat{X}_\eta(Y_n),X_n \rangle = \\
  &\sum_{i=1}^{m_n} \overline{\eta}(y_{n,i})^2 +\sum_{i=1}^r x_i^2  -2 \sum_{i,j=1}^r x_i  \overline{\eta}(y_{n,i}) \langle \V{a}_i\V{b}_i',\V{u}_{n,j}\V{v}_{n,j}' \rangle  = \\
&  \sum_{i=r+1}^{m_n} \overline{\eta}(y_{n,i})^2 + 
  \sum_{i=1}^r  \big[ \overline{\eta}(y_{n,i})^2+x_i^2- 2x_i\sum_{j=1}^r \overline{\eta}(y_{n,j}) \langle \V{a_i}\V{b_i'}, \V{u_{n,j}}\V{v'_{n,j}}\rangle \big].
\end{align*}

Assuming that $\forall y_i < \beta_+ ,\overline{\eta}(y_i) = 0$, and $rank(X)=1$, we get when $n\to\infty$: 
\begin{align*}
  \norm{\frac{1}{\mu_A}\hat{X}_\eta(Y)-X}_F^2 \aseq &  
    \overline{\eta}(y_{n,1})^2+x_1^2- 2x_1 \overline{\eta}(y_{n,1}) \langle \V{a_1}\V{b_1'}, \V{u_{n,1}}\V{v'_{n,1}}\rangle =\\
  &\overline{\eta}(y_{n,1})^2+x_1^2- 2x_1 \overline{\eta}(y_{n,1}) \langle \V{a_1},\V{u_{n,1}}\rangle \langle\V{b_1'},\V{v'_{n,1}}\rangle.
  \end{align*}

Differentiating w.r.t $\eta$ and comparing to zero we get: $\overline{\eta} = x_1 \langle \V{a_1},\V{u_{n,1}}\rangle \langle\V{b_1'},\V{v'_{n,1}}\rangle
$. 
We would like to express $\overline{\eta}$ as a function of $y_1$ instead of $x_1$.
When $n\to\infty$, the expression $\langle \V{a_1},\V{u_{n,1}}\rangle \langle\V{b_1'},\V{v'_{n,1}}\rangle$ is known.  $x_1$ can be expressed as a function of $y$ as in \ref{eq:x_of_y}. Denote $t = \frac{\mu_A\cdot x}{\sigma_B}$, then we have:

\begin{align*}
  \overline{\eta} &= x_1 \langle \V{a_1},\V{u_{n,1}}\rangle \langle\V{b_1'},\V{v'_{n,1}}\rangle 
  =\sigma_B\cdot \frac{t^4 - \beta}{\sqrt{ (t^4 + \beta t^2) (t^4 + t^2)  }}  
  =\frac{\sigma_B^2}{y} \frac{t^4-\beta}{t^2}  
  =\frac{\sigma_B^2}{y} (\sqrt{t^2-\frac{\beta}{t^2}})^2  \\
  &=\frac{\sigma_B^2}{y} \sqrt{(t^2+\frac{\beta}{t^2})^2 - 4\beta}
  =\frac{\sigma_B^2}{y} \sqrt{((t+\frac{1}{t})(t+\frac{\beta}{t}) - \beta - 1))^2 - 4\beta} 
  =\frac{\sigma_B^2}{y} \sqrt{(\frac{y}{\sigma_B}^2 -\beta-1)^2 - 4\beta}.
\end{align*}

Finally:
\[ \label{eq:eta_star}
  \eta^* = \begin{dcases} \frac{\sigma_B^2}{\mu_A y}\sqrt{( (\frac{y}{\sigma_B})^2-\beta-1)^2-4\beta}  & y \geq \sigma_B(1+\sqrt{\beta}) \\
    0  &  y < \sigma_B(1+\sqrt{\beta}) \\
  \end{dcases}
\]

Now without the assumption that $rank(X)=1$, we have:
\begin{eqnarray*}
  \norm{\hat{X}_\eta(Y)-X}_F \aseq  \sum_{i=i}^r  \eta(y_{n,i})^2+x_i^2- 2x_i \eta(y_{n,i}) \langle \V{a_i},\V{u_{n,i}}\rangle \langle\V{b_i'},\V{v'_{n,i}}\rangle.
\end{eqnarray*}

Here there are $r$ independent positive summands. The minimum of each is achieved by $\eta^*$ above \ref{eq:eta_star}. \qed

\paragraph{Proof of Theorem \ref{thm:informal_critical}.}
We first calculate $x^{critical}(\lambda^*)$.

Let $\overline{x} = \mu_A x$, and $\lambda$ a threshold operator s.t. $\lambda \geq \sigma_B(1+\sqrt{\beta})$.
Denote by $x^{est}$ be the unique solution to the intersection equality of \eqref{lemma:AMSE}:
\begin{eqnarray*}
  \overline{x}^2 = \sigma_B ^2 (( \frac{\overline{x}}{\sigma_B}+\frac{\sigma_B}{\overline{x}})(\frac{\overline{x}}{\sigma_B} +\frac{\beta \sigma_B}{\overline{x}})- ((\frac{\overline{x}}{\sigma_B})^2 - \frac{2\beta  \sigma_B ^2}{\overline{x}^2} ))
\end{eqnarray*}

The left and right hand expressions denote the zero estimator loss $L_1(\lambda_0|X) $ and   $L(\lambda|X)$ for a given $\lambda$ accordingly. For $x < x^{est} = \frac{\sigma_B}{\mu_A}c$, $L_1(\lambda_0|X) < L(\lambda|X)$ for any $\lambda$, and the proof for  $x^{critical}(\lambda^*)$ follows.

We now turn to  $x^{critical}(\eta^*)$.
According to the fundamental displacement lemma (lemma \ref{fundamental:lem} ), $x_i$ is asymptotically undetectable-- displaced to a value independent of x --  if its matching limit data singular value $y_i \leq \sigma_B( 1+\sqrt{\beta})$.
The transition between detectable and undetectable is at:
\begin{align*}
 {\overline{x_i}}= \sigma_B  \beta ^\frac{1}{4} 
\end{align*}

 Assume a shrinker $\eta$  different from the zero shrinker, s.t. for $y_i \leq \sigma_B( 1+\sqrt{\beta})$, $\eta(y_i) \neq 0$.
 The AMSE expression for any shrinker is: 
\begin{align*}
  \norm{\frac{1}{\mu_A}\hat{X}_\eta(Y)-X}_F^2 &=  \langle \frac{1}{\mu_A}\hat{X}_\eta(Y)-X,\frac{1}{\mu_A}\hat{X}_\eta(Y)-X \rangle_{HS}  =\\
  &\langle\frac{1}{\mu_A}\hat{X}_\eta(Y),\frac{1}{\mu_A}\hat{X}_\eta(Y)\rangle + \langle X_n, X_n \rangle -2\langle \frac{1}{\mu_A}\hat{X}_\eta(Y_n),X_n \rangle = \\
  &\sum_{i=1}^{m_n} \overline{\eta}(y_{n,i})^2 +\sum_{i=1}^r x_i^2  -2 \sum_{i,j=1}^r x_i  \overline{\eta}(y_{n,i}) \langle \V{a}_i\V{b}_i',\V{u}_{n,j}\V{v}_{n,j}' \rangle   \\
\end{align*}

For $\eta$ as assumed above, the expression $\sum_{i=1}^{m_n} \overline{\eta}(y_{n,i})^2$ tends to  $\infty$ with $n$: there are infinitely many noise singular values that are shrunk to a positive value. The other summands are bounded, so $L(\eta|X) < L(\eta_0|X)$ and therefore, 
 \[
   x^{critical}(\eta^*) = \frac{\sigma_B\beta^{\frac{1}{4}}}{\mu_A}\,.
    \]

\qed

\paragraph{Proof of Theorem \ref{improve:thm}.} 

We start with  $M_{\lambda^*}$.
From lemma \ref{lemma:AMSE} we have that the maximal AMSE is at the point of intersection of the AMSE expressions. This point is given by:
\begin{eqnarray*}
  x^{est} =  \sigma_B \sqrt{ \frac{ 1+\beta+\sqrt{1+14\beta+\beta^2}}{2}}  
\end{eqnarray*}

Given $\beta=1$, plug in $x^{est}$ and the calculation of $M_{\lambda^*}$ follows.

Turning to $M_{TSVD}$, 
TSVD is actually hard thresholding at the bulk edge. According to the displacement formula \ref{eq:displacement}, asymptotically it is located at $\sigma_B(1+\sqrt(\beta))$, thus signal of size slightly larger than $\sigma_B \beta ^\frac{1}{4}/\mu_A$, will be estimated as zero. Plugging in $\beta=0$, $x^{worst} = \frac{\sigma_B}{\mu_A}$. Plugging $x^{worst}$ in the first expression of the AMSE (as we assumed $x^{worst} > x(\lambda) =  \sigma_B \beta ^\frac{1}{4}/\mu_A$), the proof follows.

Turning finally to $M_{\eta^*}$, 
similarly to when we derived the optimal shrinker, we begin by focus on a single value. Considering the optimal shrinker for a rank one matrix:
\[
  \eta = x_1 \langle \V{a_1},\V{u_{n,1}}\rangle \langle\V{b_1'},\V{v'_{n,1}}\rangle,
\]
and let $t = {(\mu_A\cdot x)}/{\sigma_B}$, 
 the AMSE is:
\begin{align*}
  \norm{\hat{X}_\eta(Y)-\mu_A X}_F^2  = 
      &\eta(y_{n,1})^2+\overline{x_1}^2- 2\overline{x_1} \eta(y_{n,1}) \langle \V{a_1},\V{u_{n,1}}\rangle \langle\V{b_1'},\V{v'_{n,1}}\rangle \aseq \\
      & \overline{x}^2-\eta^2 =\\
      & \mu^2 x^2 \left( 1 - \frac {(t^4-\beta)^2}{(t^4+\beta t^2)(t^4+t^2)}  \right)
  \end{align*}

Choosing $t^2 = z$, we get
\begin{align*}
      \sigma^2 z \left( 1 - \frac {(z^2-\beta)^2}{(z^2+\beta z)(z^2+z)}  \right)
  \end{align*}

According to the assumption on $x$, $t>\beta^\frac{1}{4}$ and therefore $z>\sqrt{\beta}$.
So the entire expression is monotonically increasing with $\beta$, in $\beta\in (0,1]$.
Plugging in $\beta=1$, the simplified expression is 
\[
  \sigma^2 \left(2-\frac{1}{z}\right) \underset{z\to\infty}{\to} 2\sigma^2
\]

This is the expression for the squared asymptotic loss for every singular value, yielding MSE of $r2\sigma^2$ in total.

Using the fact that 
$\norm{\hat{X}_\lambda(Y_n)-\mu_A\cdot X_n}_F^2 = \frac{1}{\mu_A^2} \norm{\hat{X}_\lambda(Y_n)-\mu_A\cdot X_n}_F^2 $ we divide the expressions by $\mu_A^2$ and the proof follows.

\qed

\paragraph{Proof of Lemma \ref{lemma:decomp}.}
We begin with the shrinker loss.
\begin{align*}
L_\infty(\eta|\V{x})  &=   \norm{\frac{1}{\mu_A}\hat{X}_\eta(Y)-X}_F^2 =  \langle \frac{1}{\mu_A}\hat{X}_\eta(Y)-X,\frac{1}{\mu_A}\hat{X}_\eta(Y)-X \rangle_{HS}  \\&=
  \langle\frac{1}{\mu_A}\hat{X}_\eta(Y),\frac{1}{\mu_A}\hat{X}_\eta(Y)\rangle + \langle X_n, X_n \rangle -2\langle \frac{1}{\mu_A}\hat{X}_\eta(Y_n),X_n \rangle  \\&=
  \sum_{i=1}^{m_n} \overline{\eta}(y_{n,i})^2 +\sum_{i=1}^r x_i^2  -2 \sum_{i,j=1}^r x_i  \overline{\eta}(y_{n,i}) \langle \V{a}_i\V{b}_i',\V{u}_{n,j}\V{v}_{n,j}' \rangle   \\ &= 
  \sum_{i=r+1}^{m_n} \overline{\eta}(y_{n,i})^2 + 
  \sum_{i=1}^r  \big[ \overline{\eta}(y_{n,i})^2+x_i^2- 2x_i\sum_{j=1}^r \overline{\eta}(y_{n,j}) \langle \V{a_i}\V{b_i'}, \V{u_{n,j}}\V{v'_{n,j}}\rangle \big].
\end{align*}

Assuming that $\forall y_i < \beta_+ ,\overline{\eta}(y_i) = 0$,when $n\to\infty$: 

\begin{eqnarray*}
   \aseq \sum_{i=i}^r  \eta(y_{n,i})^2+x_i^2- 2x_i \eta(y_{n,i}) \langle \V{a_i},\V{u_{n,i}}\rangle \langle\V{b_i'},\V{v'_{n,i}}\rangle
\end{eqnarray*}

The AMSE decomposition for a hard threshold $\lambda$ is a result  of lemma \ref{lemma:AMSE}.

\qed

\begin{lemma} \label{lemma:AMSE}
  {\bf AMSE of hard threshold $\lambda$. }Fix the signal matrix's rank $r$ and let $\underline{x}\in \R^r$ the fixed rank singular values vector of the signal $X$. Let $\{X_n(\underline{x})\}_{n=1}^{\infty}$ , $\{A_n\}_{n=1}^{\infty}$ and $\{B_n\}_{n=1}^{\infty}$, Let $\beta$ be the sequence asymptotic ratio $X_n\in \R^{m_n\times n}(\R)$ where $\lim_{n\to\infty} m_n/n=\beta$. Let $Y_n$ matrix sequence s.t. $Y = A \bigodot X + B$ as in the basic framework \eqref{signal:def} with matching $\beta,\sigma_B,\mu_A$. 
 Let $Y$'s SVD decomposition be: 
  \begin{eqnarray} 
    Y=\sum_{i=1}^m y_i \V{u}_i \V{v}'_i
\end{eqnarray}
  
  Assume we estimate X by the hard threshold estimator with parameter k, i.e. by setting to zero $Y$'s singular values that are smaller than $k$:
  \begin{eqnarray*}
    \hat{X}_\lambda(Y) = \sum_{i=1}^{m} \eta_H(y_i;\lambda)  \V{u}_i \V{v}'_i, \text{ , where } \eta_H(y;\lambda) =  	\begin{dcases*}
		0 & $y<\lambda$,\\
		y & otherwise\\
	    \end{dcases*}
  \end{eqnarray*}
  and let $\lambda$ be a selected hard threshold s.t. $\lambda\geq \sigma_B(1+\sqrt{\beta})$. 

  Define the Asymptotic MSE:
  \begin{eqnarray*}
    AMSE(\hat{X}_\lambda(Y)) = \lim _{n\to\infty} \norm{\frac{1}{\mu_A}\hat{X}_\lambda(Y_n)- X_n}_F^2
  \end{eqnarray*} and . Then 
\begin{align} 
  AMSE(\hat{X}_\lambda , \V{x})  = \sum_{i=1}^{r} M(\hat{X}_\lambda,x_i)
\end{align}

and

\begin{align} \label{amse_:eq}
    M(\hat{X}_\lambda,x)   =
	\begin{dcases*}
	  \left(\frac{\sigma_B}{\mu_A}\right) ^2 (( \frac{\overline{x}}{\sigma_B}+\frac{\sigma_B}{\overline{x}})(\frac{\overline{x}}{\sigma_B} +\frac{\beta \sigma_B}{\overline{x}})- ((\frac{\overline{x}}{\sigma_B})^2 - \frac{2\beta  \sigma_B ^2}{\overline{x}^2} ))
&$ \overline{x}\geq x(\lambda) $\\
	x^2 & $\overline{x}< x(\lambda)$\\
	\end{dcases*}
\end{align}

where $\overline{x} = \mu_A x$.

\end{lemma}

\paragraph{Proof of lemma \ref{lemma:AMSE}} \label{proof:svht}
We'll denote the SVD of a signal matrix X, which is an element of $\{X_n\}$, by:
\[
X=\sum_{i=1}^r x_i \V{a}_i \V{b}'_i
\]

\begin{align*}
  \norm{\hat{X}_\lambda(Y)-\mu_AX}_F^2 &=  \langle \hat{X}_\lambda(Y)-\mu_AX,\hat{X}_\lambda(Y)-\mu_AX \rangle_{HS}  =\\
  &\langle\hat{X}_\lambda(Y),\hat{X}_\lambda(Y)\rangle + \langle \mu_AX_n, \mu_AX_n \rangle -2\langle \hat{X}_\lambda(Y_n),\mu_AX_n \rangle = \\
  &\sum_{i=1}^{m_n} \eta_H(y_{n,i};\lambda)^2 +\sum_{i=1}^r \overline{x_i}^2  -2 \sum_{i,j=1}^r \overline{x_i}  \eta_H(y_{n,j};\lambda) \langle \V{a}_i\V{b}_i',\V{u}_{n,j}\V{v}_{n,j}' \rangle  = \\
&  \sum_{i=r+1}^{m_n} \mu_H(y_{n,i};\lambda)^2 + 
  \sum_{i=1}^r  \big[( \mu_H(y_{n,i};\lambda)^2+\overline{x_i}^2)- 2\overline{x_i}\sum_{j=1}^r \mu_H(y_{n,j};\lambda) \langle \V{a_i}\V{b_i'}, \V{u_{n,j}}\V{v'_{n,j}}\rangle \big].
\end{align*}

According to the displacement formula \ref{eq:displacement} 
\begin{enumerate}
  \item $y_{n,r+1} \aseq \sigma_B(1+\sqrt{\beta}) < \lambda$, the leftmost term above converges almost surely to zero.
  \item When $ 0 \leq \overline{x_i} \leq \sigma_B \beta^{\frac{1}{4}}$, all that remains is $\sum_{i=1}^r {\overline{x_i}}^2 $ which proves the second case of lemma \ref{lemma:AMSE}.
\end{enumerate}

Assuming now $  \overline{x_i} \geq \sigma_B \beta^{\frac{1}{4}}$, we calculate the limit of every expression in:
\[
  \sum_{i=1}^r  \big[( \mu_H(y_{n,i};\lambda)^2+\overline{x_i}^2)- 2\overline{x_i}\sum_{j=1}^r \mu_H(y_{n,j};\lambda) \langle \V{a_i}\V{b_i'}, \V{u_{n,j}}\V{v'_{n,j}}\rangle \big].
\]
\begin{enumerate}
  \item \begin{eqnarray*}
\lim_{n\to\infty} \eta_H(y_{n,i};\lambda)^2 \aseq
\begin{dcases*}
  \sigma_B^2 ( \frac{\overline{x_i}}{\sigma_B}+\frac{\sigma_B}{\overline{x_i}})(\frac{\overline{x_i}}{\sigma_B}+\frac{\beta \sigma_B}{\overline{x_i}}) &  for $	y_{n_i} >\lambda$  ; \\
	0 & otherwise\\
\end{dcases*}
\end{eqnarray*}
\item 
  
  According to the rotation formula:

\begin{eqnarray*}\lim_{n\to\infty} \langle \V{a_i}\V{b_i'}, \V{u_{n,j}}\V{v'_{n,j}}\rangle=
  \lim_{n\to\infty} \langle \V{a_j},\V{u_{n,j}} \rangle \langle \V{b_i},\V{v_{n,j}} \rangle \aseq
\begin{dcases*}
  \frac{(\frac{\overline{x_i}}{\sigma_B})^4-\beta}{d_i\sqrt{((\frac{\overline{x_i}}{\sigma_B})^4 +\beta(\frac{\overline{x_i}}{\sigma_B})^2)((\frac{\overline{x_i}}{\sigma_B})^4 +(\frac{\overline{x_i}}{\sigma_B})^2)}}
  &  for $\overline{x_i}=\overline{x_j}$; \\
	0 & otherwise\\
\end{dcases*}
\end{eqnarray*}

where $d_i = | \{ x_j | x_j = \overline{x_i} \} | $.

\item \begin{eqnarray*}  \lim_{n\to\infty} \mu_H(y_{n,j};\lambda) \langle \V{a_i}\V{b_i'}, \V{u_{n,j}}\V{v'_{n,j}} \rangle
 \aseq
\begin{dcases*}
  \frac{\sigma_B((\frac{\overline{x_i}}{\sigma_B})^4 - \beta) }{(\frac{\overline{x_i}}{\sigma_B})^3}  & for $y_{n_i} >\lambda $; \\
	0 & otherwise\\
\end{dcases*}
\end{eqnarray*}

\item
For the last and rightmost element of the AMSE,assuming $\forall i, d_i = 1$:  
\begin{eqnarray*}
  -2\overline{x_i} \mu_H(y_{n,j};\lambda)\langle \V{a_i} ,\V{v_{n,j}} \rangle \langle \V{b_i},\V{u_{n,j}} \rangle \aseq 
     -\sigma_B^2 \left ( \frac {2x_i^2}{\sigma_B^2} - 2\beta\frac{\sigma^2}{x_i^2} \right)
\end{eqnarray*}

\end{enumerate}

Adding all expressions yields for $\overline{x_i} \geq \sigma_B \beta^{\frac{1}{4}}$ and $\overline{x_i}>x(\lambda)$: 
\begin{eqnarray*}
  	M(\hat{X}_\lambda,x) = \sigma_B ^2 (( \frac{\overline{x_i}}{\sigma_B}+\frac{\sigma_B}{\overline{x_i}})(\frac{\overline{x_i}}{\sigma_B} +\frac{\beta \sigma_B}{\overline{x_i}})- ((\frac{\overline{x_i}}{\sigma_B})^2 - \frac{2\beta  \sigma_B ^2}{\overline{x_i}^2} ))
\end{eqnarray*}

.
Which concludes the calculation by showing the first expression in the lemma. 
Using the fact that 
$\norm{\hat{X}_\lambda(Y_n)-\mu_A\cdot X_n}_F^2 = \frac{1}{\mu_A^2} \norm{\hat{X}_\lambda(Y_n)-\mu_A\cdot X_n}_F^2 $ we divide the expressions by $\mu_A^2$ and the proof follows.

 \qed

 It remains to prove Lemma \ref{fundamental:lem}. The proof follows a
strategy proposed by \cite{Nadakuditi2014}. It is convenient to break
the proof into a sequence of lemmas, which are of independent interest. 

\begin{defn}
  Let $Z$ be an $m$-by-$n$ matrix.  For $\Vu_1,\Vu_2\in\R^m$ and 
  $\Vv_1,\Vv_2\in\R^n$, define
 \begin{eqnarray}
	H ( w| \Vu_1, \Vu_2, Z) &:=& \Vu_1'(wI_n - ZZ')^{-1}\Vu_2	\\
    Q ( w| \Vv_1, \Vv_2, Z) &:=& \Vv_1'(wI_n - Z'Z)^{-1}\Vv_2	\,.
  \end{eqnarray}
\end{defn}

\begin{defn} \label{Hinf:def}
  Let $\Zc=\{Z_n\}_{n=0}^\infty$ be 
  a sequence of matrices s.t. $Z_n$ is 
  $m_n$-by-$n$ 
 and  
  $\lim_{n\to \infty} m_n/n = \beta  $.
  Let 
  $\Uc=\{U_n\}_{n=0}^\infty$ be
  any sequence of $m_n$-by-$m_n$ orthonormal matrices with columns 
  $U_n = (\Vu_{n,1},\ldots,\Vu_{n,m_n}) $. 
Similarly let
$\Vc=\{V_n\}_{n=0}^\infty$ be
  any sequence of $n$-by-$n$ orthonormal matrices with columns 
  $V_n = (\Vv_{n,1},\ldots,\Vv_{n,n}) $. 
  Define
 \begin{eqnarray}
   H_{i,j} ( w\,|\, \Uc, \Zc) &:=& \lim_{n\to\infty} H(w|\Vu_{n,i},\Vu_{n,j},Z_n)
   \label{H:eq}
   \\
   Q_{i,j} ( w \,|\, \Vc, \Zc) &:=& \lim_{n\to\infty} H(w|\Vv_{n,i},\Vv_{n,j},Z_n)
   \label{Q:eq}
  \end{eqnarray}
  assuming these limits exist almost surely.
\end{defn}

The following result is due to Bloemendal et al \cite{Bloemendal2014}:
\begin{lemma} \label{MP:lem}
  Let $\Zc=\{Z_n\}_{n=0}^\infty$ a sequence of matrices s.t. $Z_n$ 
  is $m_n$-by-$n$ and  
  $\lim_{n\to \infty} m_n/n = \beta\leq 1$. Further assume that
  $(Z_n)_{i,j} \iid \Fc$, where $\Fc$ is some 
  distribution with bounded moments, mean $0$ and variance $\sigma^2/n$. 
 Let $\Uc$ and $\Vc$ be arbitrary sequences of orthonormal matrices 
 as in Definition \ref{Hinf:def}, which are either nonrandom or 
 independent of $\{Z_n\}$. Then for all $1\leq i \leq m_n$ and 
 $1\leq j \leq n$, 
 \begin{eqnarray*}
  H_{i,j} ( w \,|\, \Uc, \Zc) 
  &\aseq& \int { \frac { d \mu_\Zc(t)}{w^2-t^2} \delta_{i,j} }  \\
  Q_{i,j} ( w \,|\, \Vc, \Zc) 
  &\aseq& \int { \frac { d\mu_\Zc(t)}{w^2-t^2} \delta_{i,j} }  
 \end{eqnarray*}
 where $ \mu_\Zc(t)$ is the density of the Mar\u{c}enko-Pastur distribution \cite{Marcenko1967} given
 by
 \begin{eqnarray} \label{MP:eq}
   \mu_\Zc(t) = \frac
   {4\sigma^4\beta - (t^2-\sigma^2-\sigma^2\beta)^2}
   {2\pi \sigma^2 \beta t}
   \mathbf{1}_{\beta_-,\beta_+}(t)\,,
 \end{eqnarray}
 where
 $ \beta_\pm = \sigma^2(1\pm\sqrt{\beta})^2 $ \,.
\end{lemma}

The next lemma shows that the matrix with entries $A_{i,j} X_{i,j}$,
with $A_{i,j}\iid (\mu_A,\sigma^2_A)$, and $X_n$ from our signal model,
is well approximated by the matrix with entries
$\E[A] X_{i,j}$.
\begin{lemma}  \label{approx:lem}
  Let $A$ be a random variable with mean $\mu_A$ and variance $\sigma_A^2$. 
  Let $\left\{ X_n \right\}$ be a matrix sequence satisfying assumptions 
  {\bf A1}--{\bf A3} (Section \ref{signal:def} in the main text).
Let $\delta_{n,1}$ be the largest singular value of the matrix $\Delta_n$ with
entries
\[
  \left( \Delta_n \right)_{i,j} = A_{i,j} X_{i,j} - \mu_A X_{i,j}\,.
\]
Then $\delta_{n,1}\aslim 0 $ as $n\to\infty$.
\end{lemma}
\begin{proof}
  See \cite{Nadakuditi2014} equations no. 35--38.
\end{proof}

The next lemma shows that adding a ``small perturbation'' to $Z_n$ does not
change the value of $H$ and $Q$ from Definition \ref{Hinf:def}.
\begin{lemma} \label{perturb:lem}
   Let $\Zc=\{Z_n\}_{n=0}^\infty$ a sequence of matrices s.t. $Z_n$ 
  is $m_n$-by-$n$ and  
  $\lim_{n\to \infty} m_n/n = \beta\leq 1$, and 
  let $\Uc$ and $\Vc$ be arbitrary sequences of orthonormal matrices 
 as in Definition \ref{Hinf:def}, which are either nonrandom or 
 independent of $\{Z_n\}$. Assume that
 for some $1\leq i \leq m_n$ and 
 $1\leq j \leq n$, 
 \begin{eqnarray*}
  H_{i,j} ( w \,|\, \Uc, \Zc) 
  &\aseq& f(w) \delta_{i,j} \qquad \text{and}\\
  Q_{i,j} ( w \,|\, \Vc, \Zc) 
  &\aseq& f(w)\delta_{i,j}\,.
 \end{eqnarray*}
 Let $\{\Delta_n\}_{n=0}^\infty$ be a sequece of matrices of the same sizes
 and assume that $\delta_{n,1}\aslim 0$ as $n\to\infty$, where 
 $\delta_{n,1}$ is the largest singular value of $\Delta_n$ ($n=1,2,\ldots$).
 Denote by $\bar{\Zc}$ the sequence of matrices $\{Z_n+\Delta_n\}$. 
 Then also
\begin{eqnarray*}
  H_{i,j} ( w \,|\, \Uc, \bar{\Zc}) 
  &\aseq& f(w) \delta_{i,j} \qquad \text{and}\\
  Q_{i,j} ( w \,|\, \Vc, \bar{\Zc}) 
  &\aseq& f(w)\delta_{i,j}\,.
 \end{eqnarray*}
\end{lemma}

\begin{proof}
  See \cite{Nadakuditi2014} equations no. 33 and 34.
\end{proof}

The next lemma in this chain of arguments is due to \cite{Benaych2012}:
\begin{lemma} \label{lemma2}
 Let $\Zc=\{Z_n\}_{n=0}^\infty$ a sequence of matrices s.t. $Z_n$ 
  is $m_n$-by-$n$ and  
  $\lim_{n\to \infty} m_n/n = \beta\leq 1$.
  Assume that $X_n$ is a sequence of matrices satistfying assumptions
  {\bf [A1]}--{\bf [A3]} from Section \ref{signal:def}.
  Define $Y_n=X_n+Z_n$ and 
  let $X_n=\sum_{i=1}^r x_i \Vut_i\Vvt_i'$
  and $Y_n = \sum_{i=1}^{m_n} y_{n,i} \Vu_{n,i}\Vv_{n,i}'$ denote their singular
  value decompositions, respectively.
  Assume that
 for any $\Uc$ and $\Vc$, which are arbitrary sequences of orthonormal matrices 
 as in Definition \ref{Hinf:def}, either nonrandom or 
 independent of $\{Z_n\}$, we have
for all $1\leq i \leq m_n$ and 
 $1\leq j \leq n$, 
\begin{eqnarray*}
  H_{i,j} ( w \,|\, \Uc, \Zc) 
  &\aseq& \int { \frac { d\mu_\Zc(t)}{w^2-t^2} \delta_{i,j} }  \\
  Q_{i,j} ( w \,|\, \Vc, \Zc) 
  &\aseq& \int { \frac { d\mu_\Zc(t)}{w^2-t^2} \delta_{i,j} } \,,
 \end{eqnarray*}
 where $ \mu_\Zc(t) $ is given by \eqref{MP:eq}.
 Then
  \begin{enumerate}
   \item For each $1\leq i\leq r$ we have 
     \begin{align} \label{y1:eq}
	\lim_{n \to \infty} y_{n_i} \aseq 
	\begin{dcases*}
	\sigma (\sqrt{(\frac{x_i}{\sigma}+\frac{\sigma}{x_i})(\frac{x_i}{\sigma}+\frac{\beta \sigma}{x_i}})) & for $x_i>\sigma  \beta ^\frac{1}{4}$ \\
	\sigma( 1+\sqrt{\beta})  & for $x_i\leq \sigma  \beta ^\frac{1}{4}$. \\
	\end{dcases*}
\end{align}
      \item  
 Let $1\leq i\leq r$ and $1\leq j \leq m_n$. 
    If $x_i>\sigma_B  \beta ^\frac{1}{4}$ and $j\leq r$, we have
    \begin{align} \label{c1:eq}
	d\cdot \lim_{n \to \infty} |\langle u_{n_i}^0,u_{n_j}\rangle|^2   \aseq 
	\begin{dcases*}
	\frac{(\frac{x_i}{\sigma})^4-\beta}{(\frac{x_i}{\sigma})^4 +\beta (\frac{x_i}{\sigma})^2}& $x_i=x_j$\\
	0 & $x_i\neq x_j$\\
	\end{dcases*}
\end{align}
and
\begin{align} \label{c2:eq}
	d\cdot \lim_{n \to \infty} |\langle v_{n_i}^0,v_{n_j}\rangle|^2   \aseq 
	\begin{dcases*}
	\frac{(\frac{x_i}{\sigma})^4-\beta}{(\frac{x_i}{\sigma})^4 +(\frac{x_i}{\sigma})^2}& $x_i=x_j$\\
	0 & $x_i\neq x_j$\\
	\end{dcases*}
\end{align}
Otherwise, both quantities converge a.s to $0$ as $n \to \infty$.
\end{enumerate}
\end{lemma}

We can finally connect the dots and prove Lemma \ref{fundamental:lem}.

\paragraph{Proof of Lemma \ref{fundamental:lem}.}  \label{fundamental:lem:proof}.
Per the lemma statement, let $Y_n=f_{A_n,B_n}(X_n)$ 
  be a matrix sequence following our signal model (Definition 
  \ref{signal:def}) with $A_n\iid(\mu_A,\sigma^2_A/n)$ and 
  $B_n\iid(0,\sigma^2_B/n)$.
  Let $\mathbf{A}_n$ be an $m_n$-by-$n$ matrix with entries 
  $(\mathbf{A}_n)_{i,j}\iid A_n$ and let
  $\mathbf{B}n$ be an $m_n$-by-$n$ matrix with entries 
  $(\mathbf{B}_n)_{i,j}\iid B_n$. We can write
  $(Y_n)_{i,j} = (\mathbf{A}_n)_{i,j} (X_n)_{i,j} + (\mathbf{B}_n)_{i,j}$.
  Letting $\Delta_n$ be the $m_n$-by-$n$ matrix with entries
  $(\Delta_n)_{i,j} = (\mathbf{A}_n)_{i,j} (X_n)_{i,j} - \mu_A (X_n)_{i,j}$
  we have 
  \[
    Y_n = \mu_A X_n + \mathbf{B}_n + \Delta_n\,.
  \]
  By Lemma \ref{approx:lem}, the top singular value $\delta_{n,1}$ of
  $\Delta_n$ satisfies 
  $\delta_{n,1}\aslim 0$ as $n\to \infty$.

  Now, choose   arbitrary sequences of orthonormal matrices 
$\Uc$ and $\Vc$
 as in Definition \ref{Hinf:def}, either nonrandom or 
 independent of $\{\mathbf{B}_n\}$.
 Let $\Bc$ denote the matrix sequence $\left\{ \mathbf{B}_n \right\}$ and 
 let $\bar{\Bc}$ denote the matrix sequence $\left\{ \mathbf{B}_n+\Delta_n
 \right\}$.
 Invoking Lemma \ref{perturb:lem} we obtain
\begin{eqnarray*}
  H_{i,j} ( w \,|\, \Uc, \Bc) 
  &\aseq& 
  H_{i,j} ( w \,|\, \Uc, \bar{\Bc}) \\
  Q_{i,j} ( w \,|\, \Vc, \Bc) 
  &\aseq& 
  Q_{i,j} ( w \,|\, \Vc, \bar{\Bc})\,. 
 \end{eqnarray*}
 for all $1\leq i \leq m_n$ and $1\leq j \leq n$.
 However,
 by Lemma \ref{MP:lem},  
for all $1\leq i \leq m_n$ and 
 $1\leq j \leq n$ we have
\begin{eqnarray*}
  H_{i,j} ( w \,|\, \Uc, \Bc) 
  &\aseq& \int { \frac { d\mu_\Zc(t)}{w^2-t^2} \delta_{i,j} }  \\
  Q_{i,j} ( w \,|\, \Vc, \Bc) 
  &\aseq& \int { \frac { d\mu_\Zc(t)}{w^2-t^2} \delta_{i,j} } \,,
 \end{eqnarray*}
 where $ \mu_\Zc(t) $ is given by \eqref{MP:eq} with 
 $\sigma\equiv \sigma_B$.
 It follows that the sequence 
 $Y_n = \mu_A X_n +(\mathbf{B}_n+\Delta_n) $
 satisfies all the assumptions of 
 Lemma  \ref{lemma2}, for
 arbitrary $\Uc$ and $\Vc$. We therefore conlude that, for $Y_n$, 
 equations \eqref{y1:eq}, \eqref{c1:eq} and \eqref{c2:eq} hold, replacing
 $\sigma$ with $\sigma_B$ and $x_i$ with $\bx_i$ ($1\leq i \leq r)$. The lemma follows.
\qed

\section{Contamination modes}
We describe here contamination modes that did not enter the main text. 

\paragraph{Basic contamination modes}
\begin{enumerate}
  \item {\bf Additive noise.} The simplest form of contamination is noise
    added to each entry. Assume $Y_{i,j} = X_{i,j} + Z_{i,j}$. The matching parameters for the model $A\bigodot X+B$ are $\mu_A = 1, \sigma_A = 0, \sigma_B = \sigma_Z$.
 
\item {\bf Missing-at-random.} Assume data with entries that are missing-at-random, where missing entries are replaced by zeros with probability (w.p.) $1-\kappa$.  Assume \[ Y_{i,j} = \begin{cases} X_{i,j} &
    \text{w.p. }\kappa\\ 0 & \text{w.p. }1-\kappa\ \end{cases}\,, \].
  \paragraph{} 
The matching parameters for the model $A\bigodot X+B$ are $\mu_A = \kappa, \sigma_A = 0, \sigma_B = 0$.   \item {\bf Outliers-at-random.}
    When some entries of $Y$ contain inordinate level of noise, the corresponding entries are said to be outliers:

    Assuming that this noise is additive, we can write 
    \[ Y_{i,j} =
      \begin{cases} X_{i,j} & \text{w.p. }\kappa\\ X_{i,j} + W_{i,j}  &  \text{w.p. } 1-\kappa
\end{cases}. \]
    With $W_{i,j}\iid (0,\tau^2)$. The matching parameters for the model $A\bigodot X+B$ are $\mu_A = 1, \sigma_A = 0, \sigma_B = \sqrt{(1-\kappa)\tau^2}$. 
 \item {\bf Multiplicative Noise.}
    Each signal entry is multiplied by a random noise distribution sample.
    $Y_{i,j} = Z_{i,j} \cdot X_{i,j} $. 
     The matching parameters for the model $A\bigodot X+B$ are $\mu_A =\mu_Z, \sigma_A = \sigma_Z, \sigma_B=0$. 
   \item {\bf Corrupt-at-random.} In some measurement processes some
entries are completely destroyed and {\em replaced} by randomly generated
noise. To distinguish this form of contamination from simple outliers, where
the original entry is not replaced, we refer to this as corruption.  
Assume \[
Y_{i,j} = \begin{cases} X_{i,j} & \text{w.p. }\kappa\\ W_{i,j}  &  \text{w.p. } 1-\kappa
\end{cases}\,.  \]  ,with $W_{i,j}\iid (0,\tau^2).$
The matching parameters for the model $A\bigodot X+B$ are $\mu_A =\kappa, \sigma_A = 0 , \sigma_B=\sqrt{(1-\kappa)\tau^2}$.   
\end{enumerate}
\paragraph{Composite contamination modes}
\begin{enumerate}
  \item{\bf Additive noise and missing-at-random.} Assume that
    additive noise has been added and then some entries were deleted and
    replaced with zeros.  Then 
    \[ 
      Y_{i,j} = \begin{cases} X_{i,j} + Z_{i,j}&
	\text{w.p. }\kappa\\  0 & \text{w.p. }1-\kappa\ \end{cases}\,.  
    \] 
   The equivalent $A$,$B$ selection is: $\mu_A = \kappa, \sigma_A = 0, \sigma_B = \sqrt{\kappa \sigma_Z^2}$.

  \item{\bf Additive noise and outliers-at-random.}         \[ 
      Y_{i,j} = \begin{cases} X_{i,j} + Z_{i,j}&
	\text{w.p. }\kappa\\ X_{i,j} + W_{i,j}  & \text{w.p. }1-\kappa\ \end{cases}\,.  
    \]  
    Where $W_{i,j}\iid (0,\tau^2)$, and $\tau$ is large compared to $\sigma_Z$.
    The equivalent $A$,$B$ selection is: $\mu_A = 1, \sigma_A = 0, \sigma_B = \sqrt{\kappa\sigma^2 +(1-\kappa)\tau^2}$.

     \item{\bf Multiplicative noise and corrupt-at-random.} \[ Y_{i,j} =
      \begin{cases} Z_{i,j} X_{i,j} & \text{w.p. }\kappa\\ W_{i,j}  & \text{w.p. }1-\kappa
    \end{cases}\,.  \]

   The equivalent $A$,$B$ selection is: $\mu_A = \kappa\cdot \mu_A, \sigma_A = 0, \sigma_B = \sqrt{(1-\kappa) \tau^2}$.

\end{enumerate}

\fi

	\section{Conclusions} \label{sec:summary} 

        Singular value shrinkage emerges as an effective method to reconstruct 
        low-rank matrices from contaminated data that is both practical and well
        understood. Through simple, carefully designed manipulation of the data
        singular values, we obtain an appealing improvement in the reconstruction
        mean square error. While beyond our present scope, following
        \cite{gavish2017optimal}, it is highly likely that the
        optimal shrinker we have developed offers the same mean square error,
        asymptotically, as the best rotation-invariant estimator based on the
        data, making it asymptotically
        the best SVD-based estimator for the target matrix.

        \section*{Acknowledgements}

        DB was supported by Israeli Science Foundation grant no. 1523/16 and 
        German-Israeli Foundation for scientific research and development 
        program no. I-1100-407.1-2015.

	\bibliography{mend_bib,hand_bib}
	\bibliographystyle{icml2017}

\end{document}